\documentclass[11pt]{article}
\usepackage{amsmath,amssymb,amsfonts} 
\usepackage{latexsym,nicefrac,bbm}

\newtheorem{theorem}{Theorem}
\newtheorem{lemma}[theorem]{Lemma}

\newtheorem{corollary}[theorem]{Corollary}

\newtheorem{remark}[theorem]{Remark}
\newtheorem{definition}[theorem]{Definition}

\newenvironment{proof}
      {\medskip\noindent{\bf Proof :}\hspace{1em}}
      {\hfill$\Box$\medskip}

\newcommand{\defeq}{\stackrel{\textup{def}}{=}}

\newcommand{\qq}{\mbox{\boldmath $q$}}
\newcommand{\xx}{\mbox{\boldmath $x$}}

\newcommand{\yy}{\mbox{\boldmath $y$}}
\newcommand{\zz}{\mbox{\boldmath $z$}}
\newcommand{\vv}{\mbox{\boldmath $v$}}
\newcommand{\ww}{\mbox{\boldmath $w$}}
\newcommand{\pp}{\mbox{\boldmath $p$}}
\renewcommand{\aa}{\mbox{\boldmath $a$}}

\newcommand{\R}{\mathbb{R}}

\newcommand{\nfrac}{\nicefrac}
\newcommand{\balpha}{\mbox{\boldmath $\alpha$}}
\newcommand{\bbeta}{\mbox{\boldmath $\beta$}}

\newcommand{\bb}{\mbox{\boldmath $b$}}
\newcommand{\cc}{\mbox{\boldmath $c$}}
\newcommand{\uu}{\mbox{\boldmath $u$}}
\newcommand{\dd}{\mbox{\boldmath $d$}}
\newcommand{\Real}{{\mbox{$\mathbb R$}}}
\newcommand{\ones}{\mbox{\boldmath $1$}}
\newcommand{\zeros}{\mbox{\boldmath $0$}}
\renewcommand{\l}{\mbox{$\lambda$}}

\newcommand{\zero}{\mbox{\boldmath $0$}}

\DeclareMathOperator*{\argmax}{arg\,max}
\newcommand{\cl}{\mbox{\rm cl}}
\newcommand{\n}{\mbox{$\eta$}}

\date{}
\title{Settling Some Open Problems on 2-Player Symmetric \\ Nash Equilibria
}
\author{Ruta Mehta \ \ \ \ \ \ \ \  Vijay V. Vazirani \ \ \ \ \ \ \ \  Sadra Yazdanbod \\ \\ 
\small{College of Computing, Georgia Institute of Technology.} \\
\texttt{\small rmehta, vazirani, syazdanb@cc.gatech.edu}}

\begin{document}
\maketitle
\thispagestyle{empty}

\begin{abstract}

Over the years, researchers have studied the complexity of several decision versions of Nash equilibrium in (symmetric)
two-player games (bimatrix games). To the best of our knowledge, the last remaining open problem of this sort
is the following; it was stated by Papadimitriou
in 2007: find a non-symmetric Nash equilibrium (NE) in a symmetric game. We show that this problem is NP-complete and the
problem of counting the number of non-symmetric NE in a symmetric game is \#P-complete.

In 2005, Kannan and Theobald 
defined the {\em rank of a bimatrix game} represented by matrices $(A, B)$ to be rank$(A+B)$ and asked whether a NE can be computed in rank 1 games in
polynomial time. Observe that the rank 0 case is precisely the zero sum case, for which a polynomial time algorithm follows
from von Neumann's reduction of such games to linear programming. In 2011, Adsul et. al. 
obtained an algorithm for rank 1 games; however, it does not solve the case of symmetric rank 1 games. We resolve this problem.

\end{abstract}

\newpage
\setcounter{page}{1}

\section{Introduction}
\label{sec.intro}

One of the major achievements of complexity theory in recent years is obtaining a fairly complete understanding of the
complexity of computing a Nash equilibrium (NE) in a two-player game in various situations; such a game can be
represented by two payoff matrices $(A,B)$, and therefore is also known as {\em bimatrix game}. Of the few
remaining open questions, we settle two in this paper regarding symmetric bimatrix games. We note that symmetry arises
naturally in numerous strategic situations. In fact, while providing game theory with its central solution concept, Nash
\cite{nash} felt compelled to also define the notion of a symmetric game and prove, in a separate theorem, that such
(finite) games always admit a {\em symmetric equilibrium}, {\em i.e.,} where all players play the same strategy.
Examples of well-known bimatrix games that are symmetric are Prisoners' Dilemma and Rock-Paper-Scissors. With the growth of
the Internet, on which typically users are indistinguishable, the relevance of symmetric games has further increased.

In a {\em symmetric game} all players participate under identical circumstances, i.e., strategy sets and payoffs.  Thus the payoff of a player $i$ depends only on the strategy, $s$, played by her and the multiset of
strategies, $S$, played by the others, without reference to their identities; moreover, if any other player $j$ were to play
$s$ and the remaining players $S$, the payoff to $j$ would be identical to that of $i$. In case of a bimatrix game $(A,B)$
such a symmetry translates to $B=A^T$. 

We first provide a brief summary of the known results. The seminal works of Daskalakis, Goldberg and Papadimitriou \cite{DGP},
and Chen, Deng and Teng \cite{CDT} proved that finding a NE in a bimatrix game, or a symmetric NE in a symmetric bimatrix game, is
PPAD-complete \cite{papa}.  Before the resolution of this long-standing question, researchers studied the complexity of
computing a NE with desired special properties. For numerous properties, these problems turned out to be
NP-hard\footnote{One exception is the problem of checking if an evolutionarily stable strategy exist in symmetric 2-player
game. This problem was recently shown to be $\sum^2_p$-hard by Conitzer \cite{conitzer}; containment in $\sum^2_p$ was shown
by Etessami and Lochbihler \cite{EL}.}, even for the case of symmetric games \cite{GZ,CS}.

In 2005, Kannan and Theobald \cite{KT} defined the {\em rank of a bimatrix game} $(A, B)$ to be rank$(A+B)$ and asked
whether a NE can be computed in rank 1 games in polynomial time. They also gave an example of a bimatrix rank 1 game that
has disconnected NE, thereby providing evidence that the problem would be a difficult one\footnote{von Stengel
\cite{stengel.r1} went further to give a symmetric bimatrix rank 1 game that has exponentially many disconnected symmetric
Nash equilibria.}.  Observe that the rank 0 case is precisely the zero sum case, for which a polynomial time algorithm
follows from von Neumann's reduction of such games to linear programming. In 2011, Adsul et. al \cite{AGMS} answered this
question in the affirmative; this appears to be the first efficient algorithm for a problem having disconnected solutions.
More recently, Mehta \cite{M} showed that for games of rank 3 or more, and for symmetric games of rank 6 or more, the
problem is PPAD-complete.

We now list the open problems we are aware of. In 2007, Papadimitriou \cite{agt.ch2} asked for the complexity of finding a
non-symmetric NE in a symmetric game.  One motivation for this problem may be the following.  In some situations, it may be
important to find a non-symmetric equilibrium in a symmetric 2-player game. As an example, under a symmetric equilibrium,
both players may exhaust the same resource, for instance if they access the same web site, and this may be undesirable.

Mehta left open the problem of determining the complexity of the following problems: finding a NE in a rank 2 game, and 
finding a symmetric NE in a symmetric game of rank 1, 2, 3, 4, or 5.

In this paper, we show that Papadimitriou's problem is NP-complete. We further show that the problem of counting the number
of non-symmetric NE in a symmetric game is \#P-complete.  We also give a polynomial time algorithm for finding a symmetric
NE in a symmetric game of rank 1. In Section \ref{sec.discussion}, we give some reasons to believe that finding a symmetric
NE in a symmetric game of rank 2 or more should not be in P.

Next, we note that given a symmetric bimatrix rank 1 game, the algorithm of Adsul et. al. \cite{AGMS} is not guaranteed to
produce a symmetric Nash equilibrium, as required in the definition of a symmetric game. Furthermore, the symmetric NE of a
symmetric rank 1 game can also be disconnected, thereby making it a difficult problem from the viewpoint of obtaining a
polynomial time algorithm.  In Section \ref{sec.tech} we give the new ideas that are needed, in addition to those of
\cite{AGMS}, to solve this problem.  For an example of a well-known game having disconnected symmetric equilibria, consider
Battle of Sexes, with appropriate payoffs so that the game is symmetric. 

Recently, McLennan and Tourky \cite{MT} gave the notion of imitation games, which simplified the existing proofs of
NP-completeness of \cite{CS} considerably for the case of symmetric games and led to even more such results. In Appendix
\ref{sec.ig} we study further properties of imitation games.

\subsection{New techniques}
\label{sec.tech}

\noindent{\bf First problem:}
Next we give an overview of the approach to show NP (\#P) hardness for computing (counting) non-symmetric NE in symmetric
games.
A quick look at the set of Nash equilibrium problems proven NP-hard suggests Non-Unique NE \cite{GZ}, i.e., whether the
given bimatrix game has two or more NE, as the most suitable problem to reduce from. Furthermore, an obvious approach is to
use the standard reduction from a bimatrix game $(A,B)$ (where both $A>0$ and $B>0$ are $m \times n$ matrices) to a
symmetric game, namely
\[
M=\begin{bmatrix}0 & A \\ B^T & 0 \end{bmatrix}
\]

For any non-zero vector $\zz\ge 0$, let $\n(z)$ denote the {\em normalized vector}, i.e., its components are
non-negative and add to 1.
It is easy to see that the Nash equilibria $(\aa, \bb)$ of $(A, B)$ are in one-to-one correspondence with the symmetric NE,
$(\n(\nfrac{\aa}{v}, \nfrac{\bb}{w}), \n(\nfrac{\aa}{v}, \nfrac{\bb}{w}))$ of $M$, where $v=\aa^TB\bb$ and $w=\aa^TA\bb$.
Furthermore, from a non-symmetric NE $((\aa, \bb), (\aa', \bb'))$  of game $(M, M^T)$, one can obtain two potential equilibria for
$(A, B)$, namely $(\n(\aa'), \n(\bb))$ and $(\n(\aa), \n(\bb'))$; indeed, one can readily confirm that they satisfy all
complementarity conditions.
However, there is a snag, namely all four of vectors $\{\aa,\bb,\aa',\bb'\}$ may not be nonzero, or one set of vectors may
be a scaled up version of the other, thereby not yielding 2 NE for the game $(A,B)$.  In fact every NE $(\aa, \bb)$ of $(A, B)$
yields a non-symmetric NE
$((\aa,\zeros),(\zeros,\bb))$ for the symmetric game $(M,{M}^T)$, and such a non-symmetric NE yields only one NE for $(A, B)$.

Let us say that a NE has {\em full support} if both players play all their strategies. We first seek a small dimensional
symmetric game which has a unique NE and moreover this NE has full support. Recently, \cite{well} showed that all symmetric
$2\times 2$ games always have pure NE, therefore leading us to consider $3 \times 3$
games. As shown in Lemma \ref{lem.D}, the game $(D, D^T)$, for the matrix $D$ specified below, has the right properties.

Next, let us define a ``blown up'' version of matrix $D$.  For every $k\in \mathbb{R}$, let $K_{n\times m}(k)$ be a matrix with $n$
rows and $m$ columns with all entries equal to $k$, and define $K$ to be the following $(1+m+n)\times (1+m+n)$ matrix:
\[
 D=\begin{bmatrix}
0 & 4& 0\\
2 & 0 & 4 \\
3 & 2 & 0\\
\end{bmatrix}
\ \ \ \ \  \ \ \ \ \ \
K=\left[\begin{array}{ccc}
0 & K_{1\times m}(4) & K_{1\times n}(0) \\
K_{m \times 1}(2) & K_{m\times m}(0) & K_{m \times n}(4)\\
K_{n \times 1}(3) & K_{n\times m}(2) & K_{n\times n}(0)
\end{array}\right]
\ \ \ \  \ \ \ \ 
G=K+\left[\begin{array}{ccc}
0 & 0 & 0\\
0 & 0 & A\\
0 & B^T & 0
\end{array}\right]
\]
Let $\aa,\aa' \in \R^m$, $\bb,\bb' \in \R^n$ and $c,c' \in \R$. Now let $\xx$ and $\yy$ be the following $(1+m+n)$-dimensional
vectors, $\xx = (c, \aa, \bb)$ and $\yy = (c', \aa', \bb')$. Define the {\em collapse} of a $(1+m+n)$-dimensional vector say
$\xx = (c, \aa, \bb)$ to be the 3-dimensional vector whose first component is $c$, the second is the sum of components of
$\aa$ and the third is the sum of components of $\bb$; we will denote this by $\cl(\xx)$. Now it is easy to see that if
$(\xx, \yy)$ is a NE of $(K, K^T)$ then $(\cl(\xx), \cl(\yy))$ must be a NE of $(D, D^T)$.  Therefore, the NE of $(K, K^T)$ must
inherit the properties of the NE of $D$, and hence all four vectors $\{\aa,\bb,\aa',\bb'\}$ must be nonzero! 

The next key idea is to insert the given game $(A, B)$ in this setup in such a way that the composite game has not only the
property established above but also captures certain essential features of the game $(A, B)$. For this we will first make
the assumption that w.l.o.g. the entries of $A$ and $B$ are positive and $< < 1$, and we construct the matrix $G$ given
above. We further define certain $3 \times 3$ matrices, $D_{\epsilon_1,\epsilon_2}$, by perturbing $D$ appropriately (see
Section \ref{sec.Hard}).

We then show that the Nash equilibria of the symmetric game $(G, G^T)$  have an ``image''  on the Nash equilibria of the
perturbed $3 \times 3$ matrix, e.g., we show that payoff of the first player, assuming $(\xx, \yy)$ is played on the
symmetric game $(G, G^T)$, is the same as the payoff of first player if $(\cl(\xx), \cl(\yy))$ is played on
$(D_{\epsilon_1,\epsilon_2} , D_{\epsilon_1',\epsilon_2'})$, for a suitable choice of $\epsilon_1,\epsilon_2,
{\epsilon_1',\epsilon_2'}$. Eventually this leads to showing that $(G, G^T)$ has a non-symmetric NE iff $(A, B)$ has at
least two NE, and moreover, the non-symmetric NE of $(G, G^T)$ are in a one-to-one correspondence with ordered pairs of NE
of $(A, B)$. These give the NP-hardness and \#P-completeness results, respectively.
\medskip

\noindent{\bf Second problem:}
An obvious approach to designing an algorithm for finding symmetric NE in rank-$1$ symmetric games is to impose symmetry in
the approach of Adsul et. al. However, this approach fails and a new approach is called for. In order to describe the
salient features of the latter, it is important to recall their approach and show where it fails.

Their approach was to start with the standard quadratic program (QP) that captures all Nash equilibria of a given bimatrix
game as optimal solutions. Since rank$(A+B) = 1$,  $A + B = \cc \dd^T$, for a suitable choice of vectors $\cc, \dd$. After making this
substitution, the objective function of the QP becomes the product of two linear forms. \cite{AGMS} replaces one of the
linear forms by a parameter $\l$, thereby getting a parameterized linear program LP$(\l)$.  They show that the optimal
solutions of this linear program, over all choices of $\l \in \R$, are precisely all NE of a certain space of rank 1 games,
i.e., $(A, \uu \dd^T - A)$, for all choices of $\uu \in \R^m$; we will denote the bimatrix game $(A, \uu \dd^T - A)$ by $(A,
\uu, \dd)$. They further show that the union of all the polyhedra defined by the constraints of LP$(\l)$, over all $\l$, is
yet another polyhedron. The one-skeleton of the latter polyhedron contains a path whose points are in one-to-one
correspondence with the optimal solutions of LP$(\l)$, over all $\l$.  Additionally, $\l$ is monotonic on this path.
Therefore, they are able find a NE of game $(A, B)$ via a binary search on this path for the ``correct'' value of $\l$.

Adapting this approach to symmetric rank one games will involve the following. We are given game $(A, A^T)$, where $A + A^T
= \cc \dd^T$, and we seek a symmetric NE $(\xx, \xx)$. Clearly, we must start with the standard QP that captures symmetric
equilibria of symmetric bimatrix games. The optimal solutions of the analogous parameterized linear program are not even NE
of games in the corresponding space of rank 1 games, i.e.,  $(A, \uu, \cc)$, for all choices of $\uu \in \R^m$.  The reason is
that this is not a space of symmetric games. 

We rectify this situation by moving to a space of symmetric bimatrix games, but of rank 2. This is made possible by the
observation that matrix $A$ can be written as the sum of a skew-symmetric matrix $K$ and the rank one matrix ${\frac {1}
{2}} \cc \dd^T$, using the fact that $\cc \dd^T$ is a symmetric matrix. Now,
replacing the vector $\dd$ by $\uu$, for all choices of $\uu \in \R^m$, we get the space of rank 2 symmetric games 
\[ \left( \left(K + {\frac {1} {2}} \cc \uu^T\right), \left(K + {\frac {1} {2}} \uu \cc^T\right)^T\right) .\]
The new LP$(\l)$ will capture all symmetric NE of this space of symmetric games. 

At this stage we introduce another idea, thereby achieving a substantial simplification. We bypass the polyhedra mentioned
above completely and reduce the problem of finding the ``correct'' $\l$ to a one-dimensional fixed-point computation in
which every fixed point is guaranteed to be rational.  Such a fixed point can be found efficiently by a binary search and
yields the ``correct'' $\l$, which in turn yields the desired symmetric NE.

\section{Preliminaries}\label{sec.prel}
A bimatrix game is a two player game, each player having finitely many pure strategies (moves).  Let $S_i,\ i=1,2$ be the set of
strategies for player $i$, and let $m\defeq|S_1|$ and $n\defeq|S_2|$.  Then such a game can be represented by
two payoff matrices $A$ and $B$, each of $m\times n$ dimension. If the first player plays strategy $i$ and the second plays $j$, then
the payoff of the first player is $A_{ij}$ and that of the second player is $B_{ij}$. 
Note that the rows of these matrices correspond to the strategies of the first player and the columns
to the strategies of second player. 

Players may randomize among their strategies; a randomized play is called a {\em mixed strategy}.  The set of mixed strategies for the
first player is $X=\{\xx=(x_1,\dots,x_m)\ |\ \xx\ge 0, \sum_{i=1}^m x_i=1\}$, and for the second player is $Y=\{\yy=(y_1,\dots, y_n)\
|\ \yy\ge 0, \sum_{j=1}^n y_j=1\}$.  By playing $(\xx,\yy) \in X \times Y$ we mean strategies are picked independently at random as per
$\xx$ by the first-player and as per $\yy$ by the second-player.  Therefore the expected payoffs of the first-player and second-player
are, respectively $ \sum_{i,j} A_{ij} x_i y_j=\xx^TA\yy\ \ \ \ \mbox{ and }\ \ \ \ \sum_{i,j}B_{ij}x_iy_j = \xx^TB\yy$.

\begin{definition}{(Nash Equilibrium \cite{agt.bimatrix})}
A strategy profile is said to be a Nash equilibrium strategy profile (NESP) if no player achieves a better payoff by a
unilateral deviation \cite{nash}. Formally, $(\xx,\yy) \in X\times Y$ is a NESP iff $\forall \xx' \in X,\
\xx^TA\yy \geq \xx'^TA\yy$ and $\forall \yy' \in Y,\  \xx^TB\yy \geq \xx^TB\yy'$. 
\end{definition}

Given strategy $\yy$ for the second-player, the first-player gets $(A\yy)_k$ from her $k^{th}$ strategy. Clearly, her best
strategies are $\argmax_k (A\yy)_k$, and a mixed strategy fetches the maximum payoff only if she randomizes among her best
strategies. Similarly, given $\xx$ for the first-player, the second-player gets $(\xx^TB)_k$ from $k^{th}$ strategy, and same
conclusion applies. These can be equivalently stated as the following complementarity type conditions,

\[
\begin{array}{ll}
\forall i \in S_1,\hspace{.06in} x_i>0 \ \ \Rightarrow\ \ & (A\yy)_i =\max_{k \in S_1} (A\yy)_k\\
\forall j \in S_2,\hspace{.06in} y_j>0 \ \ \Rightarrow & (\xx^TB)_j = \max_{k \in S_2} (\xx^TB)_k
\end{array}
\]

It is easy to get the following from the above discussion: $(\xx,\yy) \in X\times Y$ is a NE of game
$(A,B)$ if and only if the following holds, where $\pi_1$ and $\pi_2$ are scalars. 

\begin{equation}\label{eq.ne}
\begin{array}{c}
\forall i \in S_1, (A\yy)_i \le \pi_1;\ \ \ \ x_i((A\yy)_i-\pi_1)=0\\
\forall j \in S_2, (\xx^TB)_j \le \pi_2;\ \ \ \ y_j((\xx^TB)_j-\pi_2)=0\\
\end{array}
\end{equation}

Game $(A,B)$ is said to be symmetric if $B=A^T$. In a symmetric game the strategy sets of both the players are identical, i.e., $m=n$,
$S_1=S_2$ and $X=Y$. We will use $n$, $S$ and $X$ to denote number of strategies, the strategy set and the mixed strategy set
respectively of the players in such a game. A Nash equilibrium profile $(\xx,\yy)\in X\times X$ is called {\em symmetric} if $\xx=\yy$.
Note that at a symmetric strategy profile $(\xx,\xx)$ both the players get payoff $\xx^TA\xx$. Using (\ref{eq.ne}) it follows that $\xx
\in X$ is a symmetric NE of game $(A,A^T)$, with payoff $\pi$ to both players, if and only if, 
\begin{equation}\label{eq.sne}
\forall i \in S, (A\xx)_i \le \pi;\ \ \ \ x_i((A\xx)_i-\pi)=0\\
\end{equation}

We will use the above characterization to design an efficient algorithm for finding a symmetric NE of a rank-1 symmetric game
in the Section \ref{sec.symm}. Next section analyzes hardness of finding and counting non-symmetric NE in symmetric games.

\section{NP-hardness of Non-Symmetric NE in a Symmetric Game}
\label{sec.Hard}
As discussed in the introduction, existence of symmetric NE in a symmetric game is guaranteed \cite{nash}, however, a
symmetric game may not have a non-symmetric equilibrium. In this section, we show that checking existence of non-symmetric
NE in general symmetric game is NP-complete and counting such NE is \#P-complete. For the NP-completeness result, we will reduce the
problem of checking {\em non-uniqueness of Nash equilibria in bimatrix games}, which is known to be NP-complete \cite{GZ}, to checking
if {\em symmetric game has a non-symmetric equilibrium}. Refer to the second part of Section \ref{sec.tech} for an overview of the
reduction. The reduction is strong enough to also give \#P-hardness result since counting equilibria in bimatrix games is known to be
\#P-hard \cite{CS}.

We will use the definitions and notation established in Section \ref{sec.tech}. Consider the following matrix 
(also mentioned in Section \ref{sec.tech}):

\[
 D=\begin{bmatrix}
0 & 4& 0\\
2 & 0 & 4 \\
3 & 2 & 0\\
\end{bmatrix}
\]

\begin{lemma}\label{lem.D}
The symmetric game $(D,D^T)$ has a unique symmetric NE, and it has a full support.  
\end{lemma}

We will prove a stronger version of Lemma \ref{lem.D}. Given $0\le \epsilon_1, \epsilon_2 <<1$, define

\[D_{\epsilon_1,\epsilon_2}= \begin{bmatrix}
0 & 4& 0\\
2 & 0 & 4+\epsilon_1 \\
3 & 2+\epsilon_2 & 0\\
\end{bmatrix}\]

Proof of the following lemma, subsumes proof of Lemma \ref{lem.D}.

\begin{lemma}\label{lem.EF}
Consider the bimatrix game $(D_{\epsilon_1,\epsilon_2},{D_{\epsilon'_1,\epsilon'_2}}^T)$ where $0\le \epsilon_i,\epsilon'_i
<<1, i=1,2$. The game has a unique NE which has full support, and if $D_{\epsilon_1,\epsilon_2}=D_{\epsilon'_1,\epsilon'_2}$
then it is a symmetric NE.
\end{lemma}
\begin{proof}
Let $S_1$ and $S_2$ be the support sets of first and second player at a Nash equilibrium $(\xx,\yy)$. 
Recall that, we have $S_1 \subseteq \argmax_{1\leq i \leq 3} (D_{\epsilon_1,\epsilon_2}\yy)_i$ and $S_2 \subseteq \argmax_{1\leq j \leq
3} (\xx ^T{D_{\epsilon'_1,\epsilon'_2}}^T)_j$.

We will first show that $S_1=S_2=\{1,2,3\}$ for all NE of the game by discarding each of the case where $|S_i|=1$ and $|S_i|=2$ for
$i=1,2$. 
\medskip

\noindent\textbf{Case 1. $|S_1|=1$ or $|S_2|=1$.}

We will derive a contradiction for $|S_1|=1$, and the other case follows similarly. 
Suppose $S_1=\{1\}$, then $\xx=(1,0,0)$, and  
$\xx ^T{D_{\epsilon'_1,\epsilon'_2}}^T=\begin{bmatrix} 0 & 2 & 3 \end{bmatrix} \Rightarrow S_2\subset \{3\}$. Thus only best
response to $\xx$ is $\yy=(0, 0, 1)$, and it is easy to see
$((1,0,0),(0,0,1))$ is not a NE. Similarly if $S_1=\{2\}$ then to $\xx=(0,1,0)$ the only best response is $\yy=(1,0,0)$ which is not a NE,
and if $S_1=\{3\}$ then best response to $\xx=(0,0,1)$ is $\yy=(0,1,0)$, again not a NE.
\medskip

\noindent\textbf{Case 2. $|S_1|=2$ or $|S_2|=2$.} 

Again, we will derive a contradiction for $|S_1|=2$, and the other case follows similarly. 
Suppose $S_1=\{1,2\}$, then we have $\xx=(p_1,p_2,0)$ such that $p_1+p_2=1$, and  
 $$\xx ^TD_{\epsilon'_1,\epsilon'_2}^T = \begin{bmatrix}4p_2 &  2p_1&  3p_1+ (2+\epsilon'_2)p_2\end{bmatrix}.$$ Therefore, 
$S_2 = \{1,3\}$ in all NE which has $S_1=\{1,2\}$ (as $|S_2|>1$ due to Case 1), so let $\yy =(q_1,0,q_3)$.
In that case $D_{\epsilon_1,\epsilon_2}\yy=\begin{bmatrix}0 & 2q_1+(4+\epsilon_1)q_3&  3q_1\end{bmatrix}$, and therefore 
$S_1\subset \{2,3\}$, a contradiction. 
Similarly, $S_1=\{1,3\}\Rightarrow S_2=\{2,3\} \Rightarrow S_1\subset \{1,2\}$, and $S_1=\{2,3\}\Rightarrow
S_2=\{1,2\} \Rightarrow S_1\subset\{1,3\}$, contradictions.
\medskip

Thus the only possibility we are left with is $S_1=S_2=\{1,2,3\}$ for all NE of the game. Next we show that there is a unique NE with
this support. Let $((p_1,p_2,p_3),(q_1,q_2,q_3))$ be a NE of the game. Since
$S_2=\{1,2,3\}$, we have: \[4p_2=2p_1+(4+\epsilon'_1)p_3=3p_1+(2+\epsilon'_2)p_2,\ \ \ 
p_1+p_2+p_3=1 \]
Clearly, the only solution of the above equalities is
$p_1=\frac{2}{7}+\epsilon''_1,\ p_2=\frac{3}{7}+\epsilon''_2,\ p_3=\frac{2}{7}+\epsilon''_3$, where $\epsilon''_i >0,\ i=1,2,3$ depends
on $\epsilon_1$ and $\epsilon_2$. We can write similar equalities for $(q_1,q_2,q_3)$ using NE conditions for the first player, which
has a unique solution 
$q_1=\frac{2}{7}+\epsilon'''_1,\ q_2=\frac{3}{7}+\epsilon'''_2,\ q_3=\frac{2}{7}+\epsilon'''_3$. 

Since existence of symmetric NE in a symmetric game is guaranteed, if $D_{\epsilon_1,\epsilon_2}=D_{\epsilon'_1,\epsilon'_2}$, then 
the only NE of such a game is symmetric.
\end{proof}

As observed in Section \ref{sec.tech}, the well known reduction from a bimatrix game $(A,B)$ to symmetric game
$G=\left[\begin{array}{cc}0 & A \\ B^T & 0\end{array}\right]$ is not useful for our purpose.
This is because, non-symmetric NE of $(G, G^T)$ can be of the form $((\aa,\zeros),(\zeros,\bb))$, and therefore fails to produce more
than one NE of game $(A,B)$. Next we show how to circumvent this issue by constructing a suitable matrix $G$ using the game of Lemma
\ref{lem.EF} such that no component in non-symmetric NE is zero, and it relates to a unique pair of NE of game $(A,B)$. This one-to-one
correspondence gives $\#P$-hardness result as well. 

Recall the following, where $K_{c \times d}(k)$ is a $c \times d$ dimensional matrix with all entries set to $k \in \R$.

\[
K=\left[\begin{array}{ccc}
0 & K_{1\times m}(4) & K_{1\times n}(0) \\
K_{m \times 1}(2) & K_{m\times m}(0) & K_{m \times n}(4)\\
K_{n \times 1}(3) & K_{n\times m}(2) & K_{n\times n}(0)
\end{array}\right]
\ \ \ \  \ \ \ \ 
G=K+\left[\begin{array}{ccc}
0 & 0 & 0\\
0 & 0 & A\\
0 & B^T & 0
\end{array}\right]
\]

Before we go into proving our claims, we define a few terms and functions, to be used in the rest of the section.
For any non-zero, non-negative vector $\zz$, of any dimension, let $\n(\zz)$ denote the {\em normalized vector}, i.e., its components are
non-negative and add to $1$. For a matrix $M$ and two non-zero vectors $\zz_1, \zz_2 \ge 0$ of appropriate dimensions, we define,
\[
P(M;\zz_1,\zz_2)\defeq\n(\zz_1)^TM \n(\zz_2) 
\]

{\em i.e.,} the payoffs obtained by player with payoff matrix $M$ 
at strategy profile $(\n(\zz_1),\n(\zz_2))$.
Given a strategy profile $(\xx,\yy)$ of game $(G,G^T)$, where $\xx=(c,\aa,\bb)$ and $\yy=(c',\aa',\bb')$, or given NE $(\aa,\bb')$ and
$(\aa',\bb)$ for game $(A,B)$, define
\begin{equation}\label{eq.e}
\epsilon_1\defeq P(A;\aa,\bb'),\ \ \ \epsilon_2\defeq P(B;\aa',\bb),\ \ \ \epsilon'_1\defeq P(A;\aa',\bb)\ \ \
\epsilon'_2\defeq P(B;\aa,\bb')
\end{equation}

For $\xx=(c,\aa,\bb)$, recall that $\cl(\xx)\defeq (c,\sum_i a_i, \sum_j b_j)$.  
Next we show a connection between payoffs in game $(G,G^T)$ and in game $(D_{\epsilon_1,\epsilon_2},D_{\epsilon'_1,\epsilon'_2}^T)$.

\begin{lemma}\label{lem.payoff}
$P(G;\xx,\yy)=P(D_{\epsilon_1,\epsilon_2};\cl(\xx),\cl(\yy))$, and
$P(G^T;\xx,\yy)=P(D_{\epsilon'_1,\epsilon'_2}^T;\cl(\xx),\cl(\yy))$.
\end{lemma}
\begin{proof}
We will prove the first part, and the second part follows similarly.
Let $\vv=(v_1, v_2, v_3)=\cl(\xx)$ and $\ww = (w_1, w_2, w_3)=\cl(\yy)$. Then, 
$$P(G;\xx,\yy)=\xx^TG\yy=\xx^TK\yy+\xx^T\begin{bmatrix}
0 & 0 & 0\\
0 & 0 & A\\
0 & B^T & 0
\end{bmatrix}\yy = \xx^TK\yy + \aa^TA\bb'+{\aa'}^TB\bb$$
where, $\xx^TK\yy=2v_2w_1+3v_3w_1+4v_1w_2+2v_3w_2+4v_2w_3$. 
Note that $\aa=(\sum_{i \leq m} {a_i}) *\n(\aa)=v_2\n(\aa)$, and similarly $\bb=v_3\n(\bb)$, $\aa'=w_2\n(\aa')$ and $\bb'=w_3\n(\bb')$.
Thus, $\aa^TA\bb'+{\aa'}^TB\bb=\n(\aa)^TA\n(\bb')v_2w_3+ {\n(\aa')}^TB\n(\bb)v_3w_2$ and hence
$$P(G;\xx,\yy) = 2v_2w_1+3v_3w_1+4v_1w_2+2v_3w_2+4v_2w_3+\n(\aa)^TA\n(\bb')v_2w_3+
{\n(\aa')}^TB\n(\bb)v_3w_2.$$
On the other hand we have $P(D_{\epsilon_1,\epsilon_2};\cl(\xx),\cl(\yy)) =2v_2w_1+3v_3w_1+4v_1w_2+2v_3w_2+4v_2w_3+\epsilon_1v_2w_3
+\epsilon_2v_3w_2.$
Since, $\epsilon_1=P(A;\aa,\bb')=\n(\aa)^TA\n(\bb')$ and $\epsilon_2=P(B;\aa',\bb)=\n(\aa')^TB\n(\bb)$ the lemma follows.
\end{proof}

Lemma \ref{lem.payoff} implies equivalence between payoffs in games $(G,G^T)$ and
$(D_{\epsilon_1,\epsilon_2},{D_{\epsilon'_1,\epsilon'_2}}^T)$, when the strategies are mapped appropriately. Using this, next we
establish relation between their NE.

\begin{lemma}\label{lem.forward1}
If $(\xx,\yy)$ is a NE for the game $(G,G^T)$ then $(\cl(\xx),\cl(\yy))$ is a NE of game
$(D_{\epsilon_1,\epsilon_2},{D_{\epsilon'_1,\epsilon'_2}}^T)$, where $\epsilon$s are defined as per (\ref{eq.e}).  
\end{lemma}
\begin{proof}
To the contrary suppose $(\vv,\ww)=(\cl(\xx),\cl(\yy))$ is not a NE of game
$(D_{\epsilon_1,\epsilon_2},{D_{\epsilon'_1,\epsilon'_2}}^T)$. Without loss of generality (wlog) suppose first player can deviate to
$\vv'$ and get a better payoff. Let $\xx'=(v'_1,v'_2\n(\aa),v'_3\n(\bb))$. 
We claim that in game $(G,G^T)$ first player can deviate to $\xx'$ and gain, i.e., $P(G;\xx',\yy) > P(G;\xx,\yy)$, contradicting
$(\xx,\yy)$ being NE.

Note that the values of $\epsilon_1,\epsilon_2,\epsilon'_1,\epsilon'_2$ for $(\xx,\yy)$ and $(\xx',\yy)$
are the same (see (\ref{eq.e})). Then, using Lemma \ref{lem.payoff} we have:
$P(G;\xx,\yy)=P(D_{\epsilon_1,\epsilon_2};\vv,\ww) < P(D_{\epsilon_1,\epsilon_2};\vv',\ww)=P(G;\xx',\yy)$.  
\end{proof}

The next corollary follows using Lemmas \ref{lem.EF} and \ref{lem.forward1}.

\begin{corollary}\label{corr}
If $(\xx,\yy)$ is NE for the game $(G,G^T)$, where $\xx=(c,\aa,\bb)$ and $\yy=(c',\aa',\bb')$, then vectors $\aa,\aa',\bb,\bb'$ are
non-zero.  
\end{corollary}

As was our goal, the above corollary establishes non-zeroness of sub-components of a NE $(\xx,\yy)$ of the symmetric game $(G,G^T)$
that we constructed from $(A,B)$. Using this property we will show how non-symmetric NE of game $(G,G^T)$ give two distinct NE of game
$(A,B)$ and vice-versa.

\begin{lemma}\label{lem.forward2}
If $(\xx,\yy)$ is a NE for the game $(G,G^T)$, where $\xx=(c,\aa,\bb)$ and $\yy=(c',\aa',\bb')$, then $(\n(\aa),\n(\bb'))$ and
$(\n(\aa'),\n(\bb))$ both are NE for the game $(A,B)$.  
\end{lemma}
\begin{proof}
We will show that $(\n(\aa),\n(\bb'))$ is NE for the game $(A,B)$, and the proof for $(\n(\aa'),\n(\bb))$ is analogous.
By contradiction, wlog suppose the first player can change $\n(\aa)$ to $\n(\aa'')$ and get a better payoff, where  
$\sum_{1\leq i \leq m} a_i=\sum_{1\leq i \leq m} a''_i$. 
Then, we have $\cl(\xx)=\cl(\xx')$ and $\n(\aa)^TA\n(\bb')<\n(\aa'')^TA\n(\bb')$. 

Let $\xx'=(c,\aa'',\bb)$. In the proof of Lemma \ref{lem.payoff} we
showed that for $\vv=\cl(\xx)$ and $\ww=\cl(\yy)$, 
$P(G;\xx,\yy)=2v_2w_1+3v_3w_1+4v_1w_2+2v_3w_2+4v_2w_3+\n(\aa)^TA\n(\bb')v_2w_3+ {\n(\aa')}^TB\n(\bb)v_3w_2.$
Then, 

\[\begin{array}{lcl}
P(G;\xx,\yy)-P(G,\xx',\yy)& =& \n(\aa)^TA\n(\bb')v_2w_3+ {\n(\aa')}^TB\n(\bb)v_3w_2-\\
& & \ \ \ \ \  \n(\aa'')^TA\n(\bb')v_2w_3-{\n(\aa')}^TB\n(\bb)v_3w_2\\
& = &\n(\aa)^TA\n(\bb')v_2w_3-\n(\aa'')^TA\n(\bb')v_2w_3<0
\end{array}
\]
A contradiction to $(\xx,\yy)$ being a NE of $(G,G^T)$. 
\end{proof}

Next we prove the reverse of Lemmas \ref{lem.forward1} and \ref{lem.forward2}. 

\begin{lemma}\label{lem.main}
If $(\aa,\bb')$ and $(\aa',\bb)$ are NE of game $(A,B)$, and if for $\epsilon$s defined in (\ref{eq.e}), $(\vv,\ww)$
is a NE of game $(D_{\epsilon_1,\epsilon_2},D_{\epsilon'_1,\epsilon'_2})$, then
$((v_1,v_2*\aa,v_3*\bb),(w_1,w_2*\aa',w_3*\bb'))$ is a NE of game $(G,G^T)$.
\end{lemma}
\begin{proof} Let $\xx=(v_1,v_2*\aa,v_3*\bb)$ and $\yy=(w_1,w_2*\aa',w_3*\bb')$.
We will check the complementary conditions (\ref{eq.ne}) for $(\xx,\yy)$. Let $S_1$ be the support set of first player and $S_2$ be the support set of
second player w.r.t. profile $(\xx,\yy)$. In particular, we need to show: 
\[ \forall i \in S_1, (G\yy)_i=max_{1\leq i \leq m}(G\yy)_i,\ \ \ \ \ \  \forall i \in S_2, (\xx^T{G}^T)_i=max_{1\leq i \leq n}
(\xx^T{G}^T)_i.\]

We next show the first condition, and the proof for the second condition is similar. 
Let $\zz=G\yy$, then, $z_1=4w_2;\ \ \ 2\leq i \leq m+1,\ z_i=2w_1+4w_3+w_3(A\bb')_{i-1};\ \ \ m+2\leq j \leq n+m+1,\
z_j=3w_1+2w_2+w_2(B^T\aa')_{j-m-1}$.
The max among these three sets of strategies are $\max_1=4w_1$, $\max_2=2w_1+4w_3+w_3\max_{i \leq m}(A\bb')_i$ and
$\max_3=3w_1+2w_2+w_2max_{j \leq n}(\aa'^TB)_j$. If we show that $\max_1=\max_2=\max_3$ then lemma follows using the fact that
$\aa$ is a best response of the first-player against $\bb'$, $\bb$ is a best response of the second-player against $\aa'$ in game
$(A,B)$ (as $(\aa,\bb')$ and $(\aa',\bb)$ are its NE).

Since $(\vv,\ww)$ is NE for the game $(D_{\epsilon_1,\epsilon_2},{D_{\epsilon'_1,\epsilon'_2}}^T)$ where 
$\epsilon_1=\aa^TA\bb'$ and $\epsilon_2=\aa'^TB\bb$, and it is of full support (Lemma \ref{lem.EF}), we have
$4w_2=2w_1+4w_3+w_3\epsilon_1=3w_1+2w_2+w_2\epsilon_2$.
On the other hand, $(\aa,\bb')$ and $(\aa',\bb)$ being NE for the game $(A,B)$ implies
$\epsilon_1=\aa^TA\bb'=max_{1\leq i \leq m}(A\bb')_i$ and $ \epsilon_2=\aa'^TB\bb=max_{1\leq j \leq n}({\aa'}^TB)_j$.
Thus we get $\max_1=\max_2=\max_3$ as desired.
\end{proof}

The next theorem follows directly using Lemmas \ref{lem.forward1}, \ref{lem.forward2} and \ref{lem.main}.

\begin{theorem}
\label{thm.main} 
For $\xx=(c,\aa,\bb)$ and $\yy=(c',\aa',\bb')$, $(\xx,\yy)$ is a NE of game $(G,{G}^T)$ iff $(\cl(\xx),\cl(\yy))$ is a NE of 
$(D_{\epsilon_1,\epsilon_2},D_{\epsilon'_1,\epsilon'_2})$, where $\epsilon$s are defined as in (\ref{eq.e}),
and $(\n(\aa),\n(\bb'))$ and $(\n(\aa'),\n(\bb))$ are both NE of game
$(A,B)$. 
\end{theorem}

To show NP-hardness of computing non-symmetric NE in symmetric games, we need to establish connection between non-symmetric NE of game
$(G,G^T)$ and a pair of distinct NE of game $(A,B)$. Theorem \ref{thm.main} almost does the job except that no such conditions are
imposed on the NE of $(G,G^T)$ and of $(A,B)$. Next theorem achieves exactly this.

\begin{theorem}\label{thm.npc}
The symmetric game $(G,{G}^T)$ has a non-symmetric NE iff the game $(A,B)$ has more than one NE.
\end{theorem}
\begin{proof}
($\Rightarrow$)
Let $(\xx,\yy)$ be a non-symmetric NE of $(G,{G}^T)$, where $\xx=(c,\aa,\bb)$ and $\yy=(c',\aa',\bb')$.  Then, using Theorem
\ref{thm.main}, $(\n(\aa),\n(\bb'))$ and $(\n(\aa'),\n(\bb))$ are NE for the game $(A,B)$.  We need to show that these are distinct. To
the contrary suppose $\n(\aa)=\n(\aa')$ and $\n(\bb)=\n(\bb')$. Then, $\epsilon_1=\epsilon'_1$ and $\epsilon_2=\epsilon'_2$. Further,
$(\cl(\xx),\cl(\yy))$ being NE of $(D_{\epsilon_1,\epsilon_2},{D_{\epsilon'_1,\epsilon'_2}}^T)$ (Theorem \ref{thm.main}), we have
$\cl(\xx)=\cl(\yy)$ (Lemma \ref{lem.EF}). This together with $\n(\aa)=\n(\aa')$ and $\n(\bb)=\n(\bb')$ implies $c=c'$, $\aa=\aa'$ and
$\bb=\bb'$, a contradiction.

($\Leftarrow$)Let $(\aa,\bb')$ and $(\aa',\bb)$ be two different NE for the game $(A,B)$, and let $(\vv,\ww)$ be NE for the
game $(D_{\epsilon_1,\epsilon_2},{D_{\epsilon'_1,\epsilon'_2}}^T)$. Let $\xx=(v_1,v_2\aa,v_3\bb)$ and $\yy=(w_1,w_2\aa',w_3\bb')$. 
then using Theorem \ref{thm.main} $(\xx,\yy)$ is a NE of game $(G,G^T)$.
Further, $(\aa,\bb')\neq(\aa',\bb)$ implies $\xx\neq \yy$.
\end{proof}

Note that, size of $(G,G^T)$ is $O(size(A,B))$, and hence Theorem \ref{thm.npc} implies polynomial-time reduction from the 
problem of checking if a bimatrix game has more than one NE to checking if a symmetric two-player game has a non-symmetric NE. 
Since former is NP-complete \cite{GZ}, this shows NP-hardness for the latter.
Containment in NP follows since all NE of games $(G,G^T)$ are rational numbers of size polynomial in the
size (bit-length) of $G$ \cite{agt.ch2}. Thus we get the next theorem.

\begin{theorem}
Checking existence of a non-symmetric Nash equilibrium in a symmetric game is NP-complete.
\end{theorem}

The proof of Theorem \ref{thm.npc} does not indicate how the number of equilibria in the games relate to each other. We
explore this in the next theorem to show the $\#P$-completeness result.

\begin{theorem}\label{thm.SP}
There is a one-to-one correspondence between ordered pairs of NE of the game $(A,B)$ and non-symmetric $NE$ of the symmetric game
$(G,G^T)$.  
\end{theorem}
\begin{proof}
First we show that for every ordered pair of NE of the game $(A,B)$, $((\aa,\bb'),(\aa',\bb))$, there is a unique non-symmetric
NE for the game $(G,{G}^T)$. Let $(\vv,\ww)$ be NE for the game $(D_{\epsilon_1,\epsilon_2},{D_{\epsilon'_1,\epsilon'_2}}^T)$. Let
$\xx=(v_1,v_2\aa,v_3\bb)$ and $\yy=(w_1,w_2\aa',w_3\bb')$, then $(\xx,\yy)$ is a non-symmetric NE of $(G, G^T)$ (Theorem
\ref{thm.npc}).  Similarly if another NE pair $((\balpha,\bbeta'),(\balpha',\bbeta))$ also gives $(\xx,\yy)$, then using Lemma 
\ref{lem.forward2} $\balpha=\aa$, $\bbeta'=\bb'$, $\balpha'=\aa'$ and $\bbeta=\bb$. 

Next, for every non-symmetric NE $(\xx,\yy)$ of game $(G,{G}^T)$, we will construct a unique pair of distinct NE for 
game $(A,B)$. Let $\xx=(c,\aa,\bb)$ and $\yy=(c',\aa',\bb')$, then $(\n(\aa),\n(\bb'))$ and
$(\n(\aa'),\n(\bb))$ are two different NE for the game $(A,B)$ (Theorem \ref{thm.npc}). 
Suppose another symmetric NE $(\xx',\yy')$ of $(G,G^T)$, where $\xx'=(o,\pp,\qq)$ and $\yy'=(o',\pp',\qq')$, gives the same pair, i.e., 
$\n(\pp)=\n(\aa)$, $\n(\qq)=\n(\bb)$, $\n(\pp')=\n(\aa')$ and $\n(\qq')=\n(\bb')$.
Then, $\epsilon_i$ and $\epsilon'_i$, $i=1,2$ defined in (\ref{eq.e}) are the same for both $(\xx,\yy)$ and $(\xx',\yy')$. Therefore,
using Lemma \ref{lem.EF} and Theorem \ref{thm.main}, we have $(\cl(\xx),\cl(\yy))=(\cl(\xx'),\cl(\yy'))$, implying
$(\xx,\yy)=(\xx',\yy')$.
\end{proof}

The next theorem follows using the fact that counting the number of NE in a bimatrix game is $\#P$-hard \cite{CS}, and Theorem
\ref{thm.SP}. Here, containment in $\#P$ follows from the rationality of NE in bimatrix games.

\begin{theorem}
Counting the number non-symmetric equilibria in a symmetric game is $\#P$-complete.
\end{theorem}

\section{Efficient Algorithm for Symmetric Rank-$1$ Games}\label{sec.symm}
In this section we consider computing symmetric NE in symmetric constant rank games. Recall that rank
of a two player game $(A,B)$ is defined as $rank(A+B)$. 
Adsul et al.  \cite{AGMS} gave a polynomial time algorithm to compute a Nash equilibrium of a rank-$1$ bimatrix game. In
case of a symmetric game, the Nash equilibrium found by the algorithm need not be symmetric. In what follows, we design a
polynomial time algorithm to compute a symmetric Nash equilibrium in symmetric rank-$1$ games. Our algorithm is an extension
of the Adsul et al. approach.

Let $(A,A^T)$ be a symmetric game, where $A$ is an $n \times n$ square matrix.  As discussed in Section \ref{sec.prel} a mixed-strategy
$\xx \in X$ is a symmetric Nash equilibrium of game $(A,A^T)$ if and only if it satisfies (\ref{eq.sne}). Using this, the next lemma
follows.

\begin{lemma}\label{lem.scost}
If $\xx \in X$ satisfies first part of (\ref{eq.sne}) then $\xx^T A\xx-\pi\le 0$. Equality holds iff $\xx$ is a symmetric NE
of $(A,A^T)$. 
\end{lemma}

Using Lemma \ref{lem.scost} we get the following quadratic program which exactly captures the symmetric Nash equilibria of
game $(A,A^T)$.

\[
\begin{array}{lll}
\max & : & \xx^T A \xx - \pi \\
s.t. & & A\xx \le \pi;\ \ \ 
\xx\ge\zero;\ \ \ \sum_{i \in S} x_i=1
\end{array}
\]

Note that, the objective value of the above program is at most zero, and exactly zero at the optimal (Lemma \ref{lem.scost}).
If the rank of game $(A,A^T)$ is one then $A+A^T = \cc \cdot \dd^T$, where $\cc, \dd \in \Real^n$. 
Note that $\cc \cdot \dd^T$ is a symmetric matrix, and therefore, we have $\cc\cdot \dd^T=\dd \cdot \cc^T$

We will represent matrix $A$ as sum of a skew-symmetric matrix and a rank-$1$ symmetric matrix.  Let $K$ be a matrix such
that $k_{ij}=a_{ij}-\frac{c_i d_j}{2}$. This implies $A=K+\frac{1}{2}\cc \cdot \dd^T$. Since, $\cc \cdot \dd^T=\dd \cdot
\cc^T$ and $A+A^T=\cc \cdot\dd^T$, we get $K+K^T=0$. Thus, $K$ is skew-symmetric, and therefore $\zz^T K \zz=0$
for any vector $\zz \in \Real^n$. Replacing $A=K+\frac{1}{2}\cc\cdot\dd^T$ in the above quadratic program we get,

\[
\begin{array}{lll}
\max & : & \frac{1}{2}(\xx^T \cc) (\dd^T\xx) - \pi \\
s.t. & & K\xx + \frac{\cc}{2} (\dd^T \xx)\le \pi;\ \ \ 
 \xx\ge\zero; \ \ \ \sum_i x_i=1
\end{array}
\]

The above formulation is a rank-$1$ quadratic program, which is NP-hard in general.  However, we will show that it can be solved in
polynomial time using the Nash equilibrium properties.  The feasible region of the above program is linear, while the cost function is
quadratic which introduces the difficulty. The idea is to construct an LP-type formulation, using the fact that the quadratic term is a
product of two linear terms, while maintaining the fact that optimal value of the new formulation is also zero and it is achieved only
when complementarity is satisfied.  Towards this, we first replace $\dd^T\xx$ by $\lambda$ in the objective function as well as in the
inequality. This gives the following optimization problem where $\xx$ is a variable vector, and $\pi$ and $\lambda$ are scalar
variables. 

\begin{equation}\label{eq.lpl}
\begin{array}{lll}
\max & : & \frac{1}{2}\lambda (\xx^T \cc) - \pi \\
s.t. & & K\xx + \frac{\cc}{2} \lambda \le \pi;\ \ \
 \xx\ge\zero; \ \ \ \sum_i x_i=1
\end{array}
\end{equation}

\begin{lemma}\label{lem.sne1}
Let $(\xx,\l, \pi)$ be a feasible point of (\ref{eq.lpl}), then $\frac{1}{2}\lambda (\xx^T \cc) - \pi \le 0$. Equality holds iff
$x_i(K\xx + \frac{\cc}{2} \lambda - \pi)_i=0,\ \forall i \in [n]$.
\end{lemma}
\begin{proof}
Since $(\xx,\l,\pi)$ satisfies $(K\xx +\frac{\cc}{2} \lambda)_i \le \pi$ and $x_i \ge 0$, $\forall i \in [n]$, we have $x_i(K\xx
+\frac{\cc}{2} \lambda - \pi)_i \le 0$. Summing up over all $i$, and using $\sum_i x_i=1$ and $\xx^TK\xx=0$, we get $\frac{1}{2}\lambda
(\xx^T \cc) - \pi \le 0$. Further, if $x_i(K\xx
+\frac{\cc}{2} \lambda - \pi)_i = 0, \forall i \in [n]$ then $\frac{1}{2}\lambda
(\xx^T \cc) - \pi = 0$, and vice-versa.
\end{proof}

Note that, formulation (\ref{eq.lpl}) is independent of vector $\dd$, an essential for our original game. This seems 
very counter intuitive at first. However, this very property allows (\ref{eq.lpl}) to capture NE of a space of games, as established
next. Finally, we will use this rich structure to formulate one-dimensional fixed point to solve our game.
Next lemma shows that the solution set (\ref{eq.lpl}) is rich enough to contain a point for every value of $\lambda$.

\begin{lemma}\label{lem.sne2}
Given $\l \in \Real$, $\exists (\xx,\pi) \in \Real^{n+1}$ such that $(\xx,\l,\pi)$ is a solution of (\ref{eq.lpl}), and the objective
value is zero at $(\xx,\l,\pi)$.
\end{lemma}
\begin{proof}
Let $\vv = \l*\ones$, where $\ones$ is an $n$-dimensional vector with all $1$s. Now, consider a symmetric game with matrix
$Z=K+\frac{1}{2}\cc \vv^T$. Let $\zz$ be a symmetric Nash equilibrium of this game and $\pi$ be the corresponding payoff $\zz^TZ\zz$,
then $\vv^T \zz=\l$. Therefore, using (\ref{eq.sne}) we have 
\[
\forall i \in [n]:\ \ \ (K\zz + \frac{\cc}{2} \l)_i \le \pi;\ \ \ \zz_i((K\zz + \frac{\cc}{2} \l)_i-\pi)=0
\]
The first inequality implies that $(\zz,\l,\pi)$ is feasible in (\ref{eq.lpl}), and the second implies that objective value is zero at
it.  Therefore, it has to be optimal due to Lemma \ref{lem.sne1}.
\end{proof}

Lemmas \ref{lem.sne1} and \ref{lem.sne2} imply that the optimal value of (\ref{eq.lpl}) is zero, and for every $a \in \Real$ there is
an optimal solution with $\l=a$. If we substitute some value for $\l$ in (\ref{eq.lpl}), then it becomes an LP. Therefore, consider it
as a parameterized linear program $LP(\l)$. The optimal value of $LP(\l)$ for any $\l \in \Real$ is zero (due to Lemma \ref{lem.sne2}).
Therefore, solutions of (\ref{eq.lpl}) are exactly the solutions of $LP(\l),\ \forall \l \in \Real$.

\begin{remark}
Let $S$ be the set of optimal solutions of (\ref{eq.lpl}).
Since the optimal value of (\ref{eq.lpl}) is zero, it is easy to see that $(\xx,\l,\pi) \in S$ if and only if $\forall i
\in [n],\ x_i(K\xx-\frac{\cc}{2}\l-\pi)_i=0$, which requires $n$ equalities. Let us assume that the polyhedron of (\ref{eq.lpl}) 
in $(\xx,\pi,\lambda)$-space is non-degenerate. Due to the equality $\sum_i x_i=1$, the polyhedron is in $(n+1)$-dimension.
Now, since at least $n$ equalities are tight on every point of $S$, we get $S \subset 1$-skeleton of the polyhedron. Further, we can
show that every vertex of $S$ has degree two in $S$, and thus $S$ forms cycles and paths (with unbounded edges on both ends). In
addition, for any $a \in \Real$, the set of points of $S$ with $\l=a$ are exactly the solutions of $LP(a)$ which must be a closed and
convex set. Using all these properties it is easy to see that $S$ forms a single path, on which $\l$ changes monotonically.
In this paper, we bypass the polyhedron completely and give a much simpler approach in the next lemma.
\end{remark}

The result of next lemma is central to the construction of one-dimensional fixed point formulation for solving our original game
$(A,A^T)$.

\begin{lemma}\label{lem.sne3}
Given a $\l \in \Real$, if $(\xx,\pi)$ is a solution of the $LP(\l)$ then for any $\vv \in \Real^n$ satisfying $\vv^T\xx = \l$, $\xx$
is a symmetric NE of game $(Z,Z^T)$, where $Z=K+\frac{1}{2}\cc\vv^T$.
\end{lemma}
\begin{proof}
Let $(\xx,\pi)$ be a solution of $LP(\l)$, then since feasible region of $LP(\l)$ is a subset of the feasible region of (\ref{eq.lpl}),
vector $(\xx,\l,\pi)$ satisfies $(K\xx-\frac{\cc}{2}\l)_i \le \pi;\ \xx\ge \zeros; \sum_i x_i=1$. This ensures that $\xx$ is a probability
distribution vector.
Due to Lemma \ref{lem.sne2}, it also satisfies
$x_i(K\xx+\frac{\cc}{2}{\l}-\pi)_i=0, \ \forall i  \in [n]$. Setting, $\l=\vv^T\xx$, these conditions are exactly that of (\ref{eq.sne})
for strategy $\xx$ and game $(Z,Z^T)$ where $Z=K+\frac{1}{2}\cc\vv^T$.
\end{proof}

\begin{remark}
Note that both the matrices of the games constructed in Lemma \ref{lem.sne3} change with $\vv$, and $\cc \vv^T$ need not be a
symmetric matrix. Therefore, rank$(Z + Z^T) = 2$.
In the Adsul et. al. approach, the first matrix is same
in all the games, and the solutions of $LP(\l)$ are NE of a family of rank-$1$ games, which crucially uses the fact that $\yy$ need not
be same as $\xx$ (non-symmetric). For this reason, their approach is not immediately applicable for finding symmetric NE.
\end{remark}

Lemma \ref{lem.sne3} implies that if we can find a $\l$ such that the solution $(\xx,\pi)$ of $LP(\l)$ satisfies $\dd^T \xx = \l$, then
$\xx$ is a symmetric Nash equilibrium of our original rank-$1$ game $(A,A^T)$. Using this observation, consider a $1$-dimensional
correspondence $F:[d_{min},\ d_{max}]\rightarrow 2^{[d_{min},\ d_{max}]}$, where $d_{min}=\min_{i \in [n]} d_i$ and $d_{max}=\max_{i\in
[n]} d_i$.

\[
F(\l)=\{\dd^T \xx \ |\ \xx \mbox{ is a solution of } LP(\l)\}
\]

By definition we have that $\forall \l \in [d_{min},\ d_{max}]$, $F(\l)$ is non-empty (Lemma \ref{lem.sne2}) and convex. 
Now using the Kakutani fixed-point theorem, $F$ has fixed-points. Clearly, every fixed-point of $F$ gives a Nash equilibrium
by Lemma \ref{lem.sne3}, and the next theorem follows.

\begin{theorem}\label{thm.sne1}
The fixed points of $F$ exactly capture the Nash equilibria of the game $(A,A^T)$.
\end{theorem}

Since the Nash equilibrium profiles of game $(A,A^T)$ are rational vectors of size polynomial in the size of $A$ \cite{agt.bimatrix},
the fixed-points of $F$ are also rational numbers of polynomial sized (using Theorem \ref{thm.sne1}). 
Thus, one can compute an exact fixed point of $F$ in polynomial time using a simple binary search starting with the pivots $d_{min}$
and $d_{max}$, and the next theorem follows.

\begin{theorem}
The problem of computing a symmetric Nash equilibrium in a symmetric rank-$1$ game is in P.
\end{theorem}

\section{Discussion}
\label{sec.discussion}

As stated in the Introduction, the complexity of finding a symmetric NE in a symmetric game of rank 2, 3, 4, or 5 remain open.
It is easy to show that for any $k \geq 1$, the rank $k$ problem reduces to the $k$-dimensional fixed point problem, which is PPAD-hard for $k \geq 2$.
Clearly, one way of resolving these questions is to find a reduction in the reverse direction. The structure of these problems seem to indicate
that such reductions should exist and all these problems should be hard. 

\bibliographystyle{abbrv}
\bibliography{kelly,sigproc}
\appendix
\section*{Appendix}

\section{Imitation Games}\label{sec.ig}
Given an $n\times n$ matrix $A$, consider a bimatrix game $(A, I)$, where $I$ is the $n \times n$ identity matrix is called
an {\em imitation matrix}.  The name arises from the fact that when restricted to pure strategies, the second player gets a payoff of
1 if she plays the same strategy as the first player and 0 otherwise. Imitation games lead to simple reductions between games. We
illustrate this in Lemma \ref{lem1} which leads to a reduction from NE in a bimatrix game to a symmetric NE in a symmetric bimatrix
game and is well known.

\begin{lemma}
\label{lem1}
 \cite{MT}
Let $(\xx , \yy)$ be a Nash equilibrium for bimatrix game $(A, I)$ where $A$ is an $n \times n$ matrix of positive entries
and $I$ is the $n \times n$ identity matrix. Then $(\yy, \yy)$ is a symmetric Nash equilibrium for $(A, A^T)$.
\end{lemma}

\begin{proof}
Let $\alpha=\max_i {(A \yy)_i}$ and $\beta=\max_i {(I \xx)_i}$. Clearly, $\alpha > 0$ and $\beta > 0$.
Since $(\xx , \yy)$ is a Nash equilibrium for $(A, I)$,

\begin{itemize}
\item
$\forall i : \ \ x_i >0  \ \Rightarrow \  {(A \yy)_i} = \alpha$  . . .  (1).
\item
$\forall i : \ \ y_i >0  \ \Rightarrow \  {(I \xx)_i} = \beta$  . . .  (2).
\end{itemize}

It is enough show that $\forall i : \ \ y_i >0  \ \Rightarrow \  {(A \yy)_i} = \alpha$.
Clearly, $y_i > 0 \ \Rightarrow \  {(I \xx)_i} = \beta > 0$,  (by (2)) 
$\Rightarrow  x_i > 0 \ \Rightarrow \ {(A \yy)_i} = \alpha$  (by (1)).
\end{proof}

An $n \times n$ matrix will be said to be a {\em positive diagonal matrix} if each of its diagonal entries is a positive number and
each of the non-diagonal entries is zero. Observe that Lemma \ref{lem1} holds even if $I$ is replaced by an arbitrary
positive diagonal matrix. This immediately raises the questions, ``By changing the positive diagonal matrix, are we guaranteed to
get a different symmetric Nash equilibrium for $(A, A^T)$? If so, can we get all symmetric Nash equilibria in this manner?''
Theorem \ref{thm2} provides a negative answer to the first question and Theorem \ref{thm3} provides additional insights to
the second question.

\begin{theorem}
\label{thm2}
Let $(\xx , \yy)$ be a Nash equilibrium for bimatrix game $(A, I)$ and $D$ be a positive diagonal matrix with $d_i > 0$ 
in the $i^{th}$ diagonal entry. Let $s = \sum_i {x_i/d_i}$ and let $\xx'$ be the vector whose $i^{th}$ entry is 
$x_i' = x_i/(s d_i)$. Then $(\xx', \yy)$ is a Nash equilibrium for bimatrix game $(A, D)$.
\end{theorem}

\begin{proof}
Clearly  $x_i' \geq 0$ for $1 \leq i \leq n$ and $\sum_i {x_i'} = s/s = 1$, hence $\xx'$ is a probability vector. 
Clearly, $x_i > 0$ iff $x_i' > 0$. Therefore, since $\xx$ is a best response to $\yy$ in the game $(A, I)$, 
$\xx'$ is a best response to $\yy$ in the game $(A, D)$. 

Let $\max_i {(I \xx)_i} = \beta$. Clearly, $\max_i {(D \xx')_i} = \beta/s$; furthermore, the set of indices which achieve maximum
are the same in both equalities. Therefore, $y_i > 0 \ \Rightarrow \  (I \xx)_i = \beta \ \Rightarrow (D \xx')_i = \beta/s$.
Therefore, $\yy$ is a best response to $\xx'$ in the game $(A, D)$, hence proving the theorem.
\end{proof}

Lemma \ref{lem1} and Theorem \ref{thm2}  raise the question, ``Consider the bimatrix game $(A, D)$ for a fixed positive diagonal matrix $D$.  
Can one characterize the set of symmetric Nash equilibria $\yy$ of $(A, A^T)$ such that $(\xx, \yy)$ is a Nash equilibrium of
the game $(A, D)$ for some probability vector $\xx$?'' 
Theorem \ref{thm3} implies that this set consists of all symmetric Nash equilibria of $(A, A^T)$.

\begin{theorem}
\label{thm3}
Let $\yy$ be any symmetric Nash equilibrium for bimatrix game $(A, A^T)$ and $D$ be a positive diagonal matrix.
Then there is a probability vector $\xx$ such that $(\xx , \yy)$ is a Nash equilibrium for bimatrix game $(A, D)$.
\end{theorem}

\begin{proof}
We will construct a probability vector $\xx$ such that $(\xx , \yy)$ is a Nash equilibrium for bimatrix game $(A, I)$. 
Then, the needed result will follow from Theorem \ref{thm2}.

Let $S = \{i \ |  \ {(A \yy)_i} >0 \}$.
Let $\xx$ be the probability vector whose $i^{th}$ coordinate is $1/|S|$ if $i \in S$ and 0 otherwise.
Let $\max_i {(A \yy)_i} = \alpha$. Since $\yy$ is a symmetric Nash equilibrium for bimatrix game $(A, A^T)$,
$\forall i : \ \ y_i >0  \ \Rightarrow \  {(A \yy)_i} = \alpha$. Therefore, the set of indices which achieve maximum in $\xx$ are
the same as those in $\yy$, hence showing that $(\xx , \yy)$ is a Nash equilibrium for bimatrix game $(A, I)$.
\end{proof}
\end{document}